\documentclass[%
 reprint,
 amsmath,amssymb,
 aps,
]{revtex4-1}

\usepackage{graphicx}
\usepackage{dcolumn}
\usepackage{bm}
\usepackage{tensor}
\usepackage[usenames,dvipsnames]{xcolor}
 \definecolor{darkblue}{RGB}{0,0,150}
\usepackage{amssymb}
\usepackage{tensor}
\usepackage{nicefrac}
\usepackage{amsthm}
\usepackage{mathrsfs}
  \newtheorem{theorem}{Theorem}[section]

  \newtheorem{proposition}[theorem]{Proposition}
  
\usepackage[unicode]{hyperref}
\hypersetup{
  pdfinfo = {
    Author = {Alan, Morgan, Mimoso},
    Title = {Birkhoff theorem - dual null}
  },
  colorlinks = true,
  breaklinks = true,
 linkcolor = blue,
 urlcolor = blue,
 citecolor = blue,
}
\usepackage{multirow}
\usepackage{enumerate}

\usepackage[
            linecolor=orange,
            backgroundcolor=red!30,
            colorinlistoftodos]{todonotes}

\newcommand{\mO}{\ensuremath{\mathscr{O}}}
\newcommand{\ud}{\ensuremath{\mathrm{d}}}
\newcommand{\Lie}{\ensuremath{\mathcal{L}}}

\begin{document}

\preprint{APS/123-QED}

\title{Revisiting the Birkhoff theorem from a dual null point of view}

\author{Alan Maciel}
\email{alanmaciel@ime.unicamp.br}

\affiliation{Departamento de Matem\'atica Aplicada, IMECC--UNICAMP, 13083-859 Campinas, SP, Brazil}
\author{Morgan Le Delliou}
\email{delliou@lzu.edu.cn,delliou@ift.unesp.br}
\affiliation{Institute of Theoretical Physics, Physics Department, Lanzhou University, 
No.222, South Tianshui Road, Lanzhou, Gansu 730000, P R China}
\altaffiliation[Also at ]{Instituto de Astrof\'isica e Ci\^encias do Espa\c co, Universidade de Lisboa, Faculdade de Ci\^encias, 
Ed. C8, Campo Grande, 1769-016 Lisboa, Portugal}


\author{Jos\'e P. Mimoso}%

\email{jpmimoso@fc.ul.pt}

\affiliation{
 Departamento de F\'{i}sica and Instituto de  Astrof\'{i}sica e Ci\^encias do Espa\c co, Faculdade de Ci\^{e}ncias da
Universidade de Lisboa, Campo Grande, Ed. C8 1749-016 Lisboa,
Portugal}



\date{\today}

\begin{abstract}
The Birkhoff theorem is a well-known result in general relativity and it is used in 
many applications. However, its most general version, due to Bona, is almost 
unknown and presented in a form less accessible to the relativist and 
cosmologist community. Moreover, many wield it mistakenly as a simple 
transposition of Newton's iron sphere theorem. In the present work, we propose 
a modern, dual null, presentation -- useful in many explorations, including 
black holes -- of the  theorem that renders accessible most of the results of  
Bona's version. In addition, {we discuss the fluid 
contents admissible for the application of the theorem, beyond a
vacuum,
{and we demonstrate how the formalism greatly simplifies solving the dynamical equations and allows one to express the solution as a power expansion in $r$. We present a 
family of solutions that share the properties predicted by the Birkhoff theorem and  discuss the existence of trapped and antitrapped regions. The formalism manifestly shows how the type of region --- trapped or untrapped --- determines the character of the Killing vector}
}.
\end{abstract}

\pacs{Valid PACS appear here}
\maketitle


\section{Introduction}
The Birkhoff theorem \cite{Birkhoff23, 
jebsen1921general,jebsen2005general}\footnote{It was recently realised that 
Jebsen's formulation predates the contribution of Birkhoff.} states that the 
vacuum spherically
symmetric solutions are static and independent of changes in the matter
distribution sourcing the gravitational field, provided
the latter changes preserve the spherical symmetry. It is often referred as the general relativistic
counterpart of Newton's iron sphere theorem \cite{BnT87, Fleury:2016tsz}, yet 
one should be wary that this is  justified only when one is dealing with the 
gravitational field in vacuum (where the case of a cosmological constant is 
included). {Indeed as pointed out in Ref. \cite{Kim:2016osp}, Birkhoff’s theorem is  
commonly misinterpreted as  determining only the gravitational field inside a 
spherically symmetric matter distribution by its enclosed mass, while the static 
thin spherical shell surrounding a spherical central object initially proposed 
by Ref. \cite{Zhang:2012zb} demonstrates that the intermediate vacuum region's 
gravity depends also on the outer shell's mass.} {We may speculate that the similarity between the field equations of general relativity (GR) in the case of spherical symmetry, and the Newtonian equations for a central field induces this misunderstanding. However, this equivocated procedure  oversees the nonlinearity of GR which distinguishes it markedly from Newtonian gravity.}

In popular textbooks  such as  Hobson and Lasenby  \cite{HobsonEtal06},
the presentation of the theorem is, for short, that
the only vacuum solution with spherical symmetry is Schwarzschild's (although
the original formulation was that the only vacuum solution with spherical
symmetry is static \footnote{Notice that the cited authors, \cite{HobsonEtal06}, 
emphasize the fact that this contrasts with Newton's iron sphere theorem  which 
does not involve time independence.}).
Physically, the Birkhoff theorem implies that if a spherically
symmetric star undergoes strictly radial pulsations, then it cannot
propagate any disturbance into the surrounding  space \footnote{These two 
statement can also be found in other textbooks: eg. \cite{dInverno92}}.
This is obviously related to the fact that the  lowest multipolar radiation that 
propagates in general relativity is quadrupole radiation.

In the present work, we consider the Birkhoff theorem and discuss it from a 
formulation particularly fruitful for the exploration of causal structures, in 
particular of dynamical black holes, based on the behavior of the expansion of 
null congruences. This so-called dual null formalism  is of great interest, as 
it is particularly adequate and useful to deal with dynamical black holes, and  
underlies many recent results regarding both the thermodynamics of black holes, 
and more general cosmological settings \cite{Hayward:1993mw, 
Bak:1999hd,Cai:2006rs,Hayward:1997jp, Binetruy:2014ela, Maciel:2015vva}.

From an observational point of view, all the information we get from the 
Universe reaches us through null paths \cite{Bartelmann:1999yn, Ade:2013sjv, 
Fleury:2014rea}. In fact, both electromagnetic radiation and the recently 
detected gravitational waves \cite{Abbott:2016blz,TheLIGOScientific:2017qsa} 
travel on null congruences,
and hence the dual null formalism, being developed from the consideration of 
null vectors, presents itself as a particularly  appropriate tool to connect 
theoretical discussions and an understanding of the observable Universe. Although a 
known result, the understanding of the  Birkhoff theorem can 
benefit from those modern tools. 

We show in this work that the dual null formalism   allows us to extend the 
Birkhoff theorem to more general geometric frameworks, such as planar or 
cylindrically symmetric spacetimes, as  well as ADS/CFT settings 
\cite{Maldacena:1997zz, birmingham1999topological, hartnoll2008building}. 
Although the generalization of the Birkhoff theorem to the latter geometrical cases 
has been previously obtained in the literature --- see Stephani 
et al. \cite{StephaniEtal03} and 
references therein -- this fact is widely ignored and was 
derived in a different way in the present work. Moreover, 
besides characterizing {naturally} the admissible matter models
that are compatible with the theorem, the dual null formalism allows us to 
find all the solutions for sources that 
can be expressed as a power series on $r$ in a simple way. Finally, the dual null formalism manifestly shows that the character of the theorem's additional Killing vector, timelike or spacelike, naturally follows from the type of region it applies, trapped or untrapped.

The outline of the present work is as follows. In Sec. \ref{sec:level1} we briefly review the literature in connection with the Birkhoff theorem. This will enable us to situate our work with regard to the alternative approaches to its derivation, as well as to some of the efforts pursued in the literature to generalize it. In Sec. \ref{Sec:BT_2null}
we present the dual null formalism used in this work, and develop the new proof of the Birkhoff theorem.  In Sec. \ref{Sec:Beyond}, we discuss the symmetry requirements that are actually needed, and obtain the most general  admissible matter models which are compatible with the theorem. Finally we give a brief discussion of our results in Sec. \ref{Conclusion}.


A quick remark on the notation: In most instances we use the abstract index notation as in Wald's textbook
\cite{wald-book}.
However,  we swap to the intrinsic mathematical notation (without indices) when it is convenient. The translation from the two notations can  be readily made by the use of the base vectors and 1-forms. For a vector $V^a$ we write $V^a = V^\mu \partial_\mu = \partial_V$ and for a 1-form $\omega_a = \omega_\mu \ud x^\mu$. If $\omega_a = \partial_a f$, for a scalar function $f$, then $\omega = \ud f = \partial_\mu f \ud x^{\mu}$.

\section{\label{sec:level1}General formulations of the Birkhoff theorem}

A generalized and geometrically minded version of the Birkhoff theorem was put forward by Goenner \cite{Goenner70}, pointing out that 
the  theorem relies on the existence of a three-parameter group of (global) isometries with two-dimensional non-null orbits and of an additional Killing vector associated with a $G_4$ group of motions \footnote{(Goenner's abstract) The Einstein tensors of metrics having a three-parameter
group of (global) isometries with two-dimensional non-null orbits G3(2,s/t)
are studied in order to obtain algebraic conditions guaranteeing an
additional normal Killing vector. It is shown that Einstein spaces
with G3(2,s/t) allow a G4 group. A critical review of some of the literature
on Birkhoff's theorem and its generalizations is given.}. The
Birkhoff theorem for spherically symmetric vacuum solutions and
the Taub theorem for plane-symmetric vacuum solutions were both generalized
to vacuum solutions with conformal symmetries. In particular, it was proved
that any conformally spherically (respectively,  plane-) symmetric
vacuum solution to the Einstein equations must be the Schwarzschild (respectively,
either Taub-Kasner or flat) solution. 

 { Upgraded versions of the theorem can be found 
in Refs.
\cite{StephaniEtal03,1973CMaPh..33...75B,Goenner70,bronnikov1980generalisation}}
, and the most evolved 
phrasing for this geometric approach is due to Bona \cite{Bona(1988)}, for metrics 
that are conformally reducible, that is ${g} = Y^2 \hat{g}$, where $\hat{g}$ is 
reducible as the metric of a direct product spacetime. Let 
\begin{gather}
\ud s^2= Y^2(x^C) \left(\gamma_{AB}\ud x^A \ud x^B + h_{\alpha \beta}\ud y^\alpha \ud y^\beta\right) \,, \label{eq:reducible-metric}
\end{gather}
where $h_{\alpha \beta}$ and $y^\alpha$ are a two-dimensional metric and a 
coordinate system, respectively, on the two-dimensional orbits $O_2$ of $G_3$. Analogously, 
$\gamma_{AB}$ and $x^A$ the corresponding metric and coordinates on $V_2$, which 
is the orthogonal submanifold to $O_2$ according to $g$. Bona's statement of the
Birkhoff theorem is

\begin{quotation}\label{BTheorem}
Theorem \cite{Bona(1988)}: Metrics with a group $G_{3}$ of motions on non-null orbits
$O_{2}$ and with Ricci tensors of type $[(11)(1,1)]$ and $[(111,1)]$
admit a group $G_{4}$ provided that $\ud Y \neq0$ . 
\end{quotation}
emphasizing the requirement  $\ud Y\neq0$  and the appropriate Segr\'e types \cite{Bona(1988)}.

Other attempts at generalizing the Birkhoff theorem can be found in the literature. 
{Generalization to higher dimensions was achieved by K. A. 
Bronnikov and V. N. 
Melnikov \cite{Bronnikov:1994ja} and a thorough discussion on the 
relationship between manifold dimensionality and the existence of 
Birkhoff-like theorems was made by H.-J. Schmidt \cite{Schmidt:2012wj}}. 
R. Goswami and G.F.R. Ellis 
\cite{Goswami:2012jf,Goswami:2011ft,Ellis:2013dla} have investigated the 
possibility of extending it by analyzing whether the theorem remains approximately 
true both for an approximately spherical vacuum solution \cite{Goswami:2012jf}, 
and also for an approximately vacuum configuration \cite{Goswami:2011ft}. 
They 
resort to the analysis of perturbations with the 1+1+2 formalism developed by 
Clarkson \cite{Clarkson:2007yp}. The difficulties associated with this pragmatic 
line of research stem from the need to remain in the neighborhood of the vacuum 
spherically symmetric models, and, thus of defining the conditions that guarantee 
the existence of such neighborhood.  


Following a diverse path, Hern\'andez-Pastora \cite{HernandezPastora:2009zz} 
pursued an
attempt to get a relationship between the spherical symmetry and the multipole structure of the so-called monopole solution.

The Birkhoff theorem was also investigated in connection with conformal rescaling 
\cite{1984PhRvL..53..315R}, with the possibility of extending it to modified 
theories of gravity \cite{Deser:2005gr,Faraoni:2010rt,Nzioki:2013lca}, 
{with 
different hypotheses, as in Ref. \cite{Bronnikov:2016dhz} --- where the {key condition (for Bona)} $\ud Y \neq 0$ 
is abandoned and the theorem still applies under some additional conditions on 
the matter sources ---} and with regard to many 
other features 
\cite{Deser:2004gi,Faraoni:2017uzy,Faraoni:2018mes,2010CQGra..27t5024F}.

\section{The Birkhoff theorem in Dual null formalism}\label{Sec:BT_2null}

{In this section, we briefly present the main tools of the dual 
null formalism and apply them to prove and discuss the necessary conditions for 
the validity of the Birkhoff theorem. It is worth pointing out that the dual 
null formalism is distinct from operating in null coordinates, as it deals with optical scalars related to null congruences. Such 
quantities are independent of coordinate choice and can be analyzed in any 
coordinate set. Null coordinates are useful in order to represent and compute 
more simply the relevant quantities of the dual null formalism, and we take 
advantage of this in the following.}

\subsection{Spherically symmetric spacetimes and dual null formalism}

In dual null coordinates, any spherically symmetric metric can be 
parametrized as
\begin{gather}
\ud s^2 = - e^{f} \left( \ud u \, \ud v + \ud v \, \ud u\right) + r^2 \left( \ud \theta^2 + \sin^2 \theta \ud \phi^2 \right)\,, \label{eq:metric-dualnull}
\end{gather} 
where $f = f(u,v)$, $r= r(u,v)$ and we omit the tensor product symbol $\otimes$ for short. Metric \eqref{eq:metric-dualnull} is of the form \eqref{eq:reducible-metric} for $Y = r(u,v)$, $\gamma_{AB}\,\ud x^A \, \ud x^B = -\frac{e^f}{r^2}\,\ud u \, \ud v$ and $h_{\alpha \beta} \,\ud y^\alpha \, \ud y^\beta = \ud \theta^2 + \sin^2  \theta \, \ud \phi^2$.

The coordinates in Eq.~\eqref{eq:metric-dualnull} are also a codimension-two foliation of the spacetime. The orbits  of the $G_3$ group, here the group of rotations in three dimensions, are two-dimensional spheres corresponding to $O_2$. Each two-dimensional sphere is characterized by the pair $(u\,,v)$, that are the coordinates on $V_2$.

The null coordinates on $V_2$ are not unique. Hence, by making a coordinate change of the form $(u\,, v) \to (U\,, V)$:
\begin{gather}
u \to U(u)\, \quad v \to V(v) \, , \label{eq:gauge}
\end{gather}
with $U'(u) > 0$ and $V'(v) > 0$ for all $u, \, v$, in order to not reverse the orientation of the new coordinates. We obtain a new pair of dual null coordinates:
\begin{gather}
\ud s^2 = - e^{F(U,V)} (\ud U \ud V +\ud V \ud U) + r^2 (U, V) \ud \Omega^2\,,
\end{gather}
with
\begin{gather}
F(U,V) = f(u(U), v(V)) - \ln U'(u(U)) - \ln V'(v(V)) \,. \label{eq:Fredefined}
\end{gather}

Let $k^a$ be a null vector field  orthogonal to the orbits of the coordinates 
$\theta$ and $\phi$ everywhere in the spacetime. We define its expansion 
$\Theta_{(k)}$ as the relative variation of the area form on the orthogonal 
spheres when transported along the integral curves of $k^a$:
\begin{gather}
\Theta_{(k)}= \frac{\Lie_k(r^2 \sqrt{\det h})} {r^2 \sqrt{\det h}} = \frac{2}{r} 
 k^a \partial_a r\,, \label{eq:expansion-definition}
\end{gather}
where $\Lie_k$ is the Lie derivative with respect to $k^a$.
Using the coordinate base vectors $\partial_u$ and $\partial_v$ we define 
the two null expansions related to our coordinates in 
Eq.~\eqref{eq:metric-dualnull}:
\begin{gather}\label{eq:nullexpansion}
\Theta_{(u)} = \frac{2}{r}\partial_u r \, , \quad \Theta_{(v)} =\frac{2}{r} \partial_v  r \, ,.
\end{gather}
The null expansions transform under \eqref{eq:gauge} as
\begin{gather}
\Theta_{(u)} \to U'(u) \Theta_{(u)} \, , \quad \Theta_{(v)} \to V'(v) 
\Theta_{(v)}\, \label{eq:gaugetransform}.
\end{gather}

We see that the value of the null expansions depends on the coordinate choice, 
but their sign and the locus where they vanish are not. Based on this,  we may 
classify each sphere in the spacetime as
\begin{itemize}
\item regular, normal or untrapped, if $\Theta_u \Theta_v < 0$;
\item trapped or future trapped, if $\Theta_u \Theta_v > 0$ and $\Theta_u < 0$;
\item antitrapped or past trapped, if $\Theta_u \Theta_v > 0$, and $\Theta_u > 0$.
\item marginal, if $\Theta_u \Theta_v$ = 0.
\end{itemize}

This classification has been an important tool in the study of black hole 
physics, especially in the case of dynamical solutions (see, for example, 
Refs. \cite{Hayward:1993mw,Hayward:1994bu,Senovilla:2011fk} and references 
therein, for motivation and applications of this formalism).

Let the basis forms related to the coordinates $(u\,, v)$ be denoted $\ud u$ 
and $\ud v$. Since $\Theta_{(u)} \ud u$ is invariant under change of coordinates 
\footnote{Or, equivalently, $\Theta_{(u)}$ transforms as a component of a 
covariant vector.}, we may build a 1-form $\mathcal{K}_a$ called the mean 
curvature form as
\begin{gather}
\mathcal{K}_{a} = \Theta_{(u)} \partial_{a} u + \Theta_{(v)} \partial_{a} v\,, \label{eq:mean-curvature}
\end{gather}
where $\partial_a u$ and $\partial_a v$ are the abstract index notation version of $\ud u$ and $\ud v$, respectively.

With the aid of the mean curvature form, we are able to express simply the null expansion respective to any null vector field by just contracting it to $\mathcal{K}_a$:
\begin{gather}
\Theta_{(k)}= k^a \mathcal{K}_a\,,
\end{gather}
\noindent
for any $k^a$ null and orthogonal to $O_2$. We may also generalize the 
definition for any vectors in $V_2$, be it time- or spacelike, by defining 
what we call the 2-expansion, in order to distinguish it from the usual 
expansion defined as the divergence of timelike vector fields, as was made in Ref.
\cite{Maciel:2015vva}  in order to deal with the separation between collapse and cosmological expansion (see \cite{Lasky:2007ky,Lasky:2006zz,LaskyLun06+,LaskyLun07,LaskyLun06b, Delliou:2009dm,Delliou:2009dc,LeDMM09a,Delliou:2013xra,Mimoso:2011zz,Mimoso:2013iga,MLeDM09}) and as a tool to define dynamical universal horizons in Ref. \cite{Maciel:2015ypv} (see \cite{Barausse:2011pu,Blas:2011ni,Berglund:2012fk,Berglund:2012bu,Bhattacharyya:2014kta,Lin:2014ija,Lin:2014eaa,Lin:2016myf,Saravani:2013kva,Tian:2015vha}). Let be $X^a$ any vector orthogonal to the orbits $O_2$; then its 2-expansion, denoted $\Theta_{(X)}$ is defined as
\begin{gather}
\Theta_{(X)} = X^a \mathcal{K}_a\,.
\end{gather}

For the expansions of the null coordinate basis, the Raychaudhuri equations are written as
\begin{subequations} \label{eq:Raychaudhuri}
\begin{gather}
\Lie_u \Theta_{(u)} -  \Theta_{(u)} \, \partial_u f + \frac{\Theta_{(u)}^2}{2} + R_{uu}=0\, ,\label{eq:Rayuu}\\ 
\Lie_v \Theta_{(v)} -  \Theta_{(v)} \, \partial_v f + \frac{\Theta_{(v)}^2}{2} + R_{vv}=0\,\label{eq:Rayvv} ,
\end{gather}
\end{subequations}

The $uv$ component of the Einstein tensor may be written in terms of the null expansions as
\begin{gather}
 G_{uv}=\Lie_v \Theta_{(u)} + \Theta_{(u)} \Theta_{(v)} + \frac{e^f}{r^2}\,.\label{eq:Guv}
\end{gather}
Note that since $\partial_u$ and $\partial_v$ are coordinate base vectors, they commute, and then $\Lie_u \Theta_{(v)} = \Lie_v \Theta_{(u)}$.

Equations~\eqref{eq:Raychaudhuri} and \eqref{eq:Guv} together with Einstein's equation
\begin{gather}
G_{ab} = T_{ab} \,, \label{eq:einstein}
\end{gather}
for a given energy-momentum tensor $T_{ab}$ capture the full 
dynamics of the problem and completely determine a spherically symmetric 
solution. 

\subsection{Properties of vacuum spacetimes} \label{sec:vacuum}

Until this point, the only hypothesis made on the spacetime was spherical 
symmetry. In this section we also assume that it satisfies Einstein's equation 
in vacuum in an open domain $\mathcal{D}$, of the form $\mathcal{D}_2 \times 
O_2$ where $\mathcal{D}$ is the image under the coordinate map $(u,v)$ of an 
open domain of $\mathbb{R}^2$. This domain $\mathcal{D}$ can be described as a 
spherical shell with finite thickness that lasts for some finite time interval.

On $\mathcal{D}$, $R_{ab} = G_{ab} = 0$. The full dynamics are determined in terms of the null expansions by the three equations below
\begin{subequations}\label{eq:vacuumeqs}
\begin{gather}
\Lie_u \Theta_{(u)} -\Theta_{(u)} \, \partial_u f + \frac{\Theta_{(u)}^2}{2} = 0\, ,\label{eq:thetau-equation}\\
\Lie_v \Theta_{(v)} -  \Theta_{(v)} \,\partial_v f  +\frac{\Theta_{(v)}^2}{2} = 0\, , \label{eq:thetav-equation}\\
\Lie_v \Theta_{(u)} + \Theta_{(u)} \, \Theta_{(v)} + \frac{e^f}{r^2} = 0\,. \label{eq:thetauv-equation}
\end{gather}
\end{subequations}

Using Eqs.~\eqref{eq:vacuumeqs}, we can deduce several general results valid for vacuum solutions that we present in the following.

\begin{proposition}
Given the hypotheses above, there exists a pair of dual null coordinates $(U, V)$ such that $|\Theta_{(U)} (U,V)| = |\Theta_{(V)} (U, V)|$ in $\mathcal{D}$. \label{prop:equal-expansions}
\end{proposition}
\begin{proof}
We can rewrite Eq.~\eqref{eq:thetau-equation} as
\begin{gather}
\partial_u \Theta_{(u)}  - \Theta_{(u)}\,  \partial_u f + \frac{\Theta_{(u)} ^2}{2} = 0 \Rightarrow \nonumber\\
r^{-1} \partial_u (\Theta_{(u)} r ) - \Theta_{(u)} \,\partial_u f = 0 \Rightarrow \nonumber\\
\frac{ \partial_u (\Theta_{(u)} r)}{\Theta_{(u)} r}  = \partial_u f \Rightarrow \nonumber\\
r\Theta_{(u)}  = C_1 (v) e^f \,,
\end{gather}
where $C_1(v)$ is an arbitrary nonvanishing function that comes from the integration in $u$. Repeating the same procedure with Eq.~\eqref{eq:thetav-equation} we obtain
\begin{gather}
r\Theta_{(v)} = C_2 (u) e^f\,,
\end{gather}
where $C_2(u)$ is also an arbitrary function.
For any functions $C_1$ and $C_2$, we can make a gauge transformation $(u, v) \rightarrow (U, V)$ of the form Eqs.~\eqref{eq:gauge} with the choice:
\begin{gather}
U(u) = \int^u |C_2 (s)| \ud s\, \nonumber\\
V(v) = \int^v |C_1 (s)| \ud s\,,
\end{gather}
noting that $U'(u)\,, \,\, V'(v) > 0$, for all $u, v$, in order to have a well-behaved coordinate transformation.  We obtain
\begin{gather}
r\Theta_{(U)} =|C_2(u)| C_1 (v) e^f \,, \nonumber\\
r\Theta_{(V)} =| C_1 (v)| C_2 (u) e^f\,. \label{eq:UVequations}
\end{gather}

Dividing  Eqs.~\eqref{eq:UVequations} by each other, we obtain the wished result.

\end{proof}

This result shows that there exists one special set $(U,V)$ of dual null 
coordinates in vacuum spherically symmetric spacetimes for which the two null 
expansions have the same absolute value at each event on $\mathcal{D}$. Note that 
this special set of dual null coordinates is unique up to a constant rescaling. 

The next proposition shows that this special pair of null coordinates is useful to 
reduce the dynamical equations to equations on only one independent coordinate.

\begin{proposition} \label{prop:killingvector}
Let  $\Theta_{(U)}  +\Theta_{(V)} = 0$ ($\Theta_{(U)} - \Theta_{(V)}=0$) and the new coordinates $\chi_{\pm} =\frac{1}{2} \left(U \pm V\right)$. We denote with $\partial_{\pm}$ the derivatives with respect to $\chi_{\pm}$.

Then
\begin{enumerate}[i.]
\item $\partial_+ \Theta_{(U)} = \partial_+ \Theta_{(V)} = 0$ ($\partial_- \Theta_{(U)} = \partial_- \Theta_{(V)} = 0$).\label{dchitheta}

\item $\partial_+ r(U,V) = 0$ ($\partial_- r(U, V) =0$).
\item If $\Theta_{(U)} \neq 0$, then $\partial_+ f = 0$ ($\partial_- f = 0$).
\item If $\Theta_{(U)} + \Theta_{(V)}=0$ ($\Theta_{(U)} - \Theta_{(V)}=0$), then $\partial_+$ ($\partial_-$) is a Killing vector.
\end{enumerate} 
\end{proposition}
\begin{proof}
First, we remark that $\partial_\pm = \partial_U \pm \partial_V$.
\begin{enumerate}[i.]

\item Let $\mO = \Theta_{(U)} = \mp \Theta_{(V)}$. Since $\partial_U$ and $\partial_V$ commute, we have
\begin{gather}
\partial_V \Theta_{(U)} = \partial_U \Theta_{(V)} \Rightarrow \nonumber\\
\partial_V \mO = \mp \partial_U \mO \Rightarrow 
\left( \partial_U \pm \partial_V \right) \mO =0\,,
\end{gather} 
where the sign choice depends directly on the choice in $\Theta_V \pm \Theta_U = 0$.

\item Consider the case $\Theta_{(U)}+ \Theta_{(V)}=0$, the other case being similar. By Eq.~\eqref{eq:nullexpansion}, $\mO = \partial_U \ln r^2 = - \partial_V \ln r^2$. Considering that $\partial_+$ commute with both $\partial_U$ and $\partial_V$ and item~(\ref{dchitheta}), we have
\begin{gather}
0 = \partial_+ \partial_U \ln r^2 = \partial_U \partial_+ \ln r^2 \Rightarrow \partial_+ \ln r^2 = C_1 (V)\, , \\
0 = \partial_+ \partial_V \ln r^2 = \partial_V \partial_+ \ln r^2 \Rightarrow \partial_+ \ln r^2 = C_2 (U)\, .
\end{gather}

Since $C_1(V) = C_2(U)$, they are constant. We can then write
\begin{gather}
\ln r^2 = C \chi_+ + H(\chi_-)\,,\label{eq:Cplush}
\end{gather}
with $H$ an arbitrary function and $C$ an arbitrary constant. Computing the expansions for $U$ and $V$ using Eq.~\eqref{eq:Cplush}, we obtain
\begin{gather}
\Theta_{(U)} = \frac{C+H'}{2} = - \frac{C - H'}{2} = - \Theta_{(V)} \Rightarrow
C = 0\,,
\end{gather}
which implies that the null expansions do not depend on $\chi_+$.

\item Adding Eq.~\eqref{eq:thetau-equation} and minus \eqref{eq:thetav-equation} and writing the expansion in terms of $\mO$, we obtain:
\begin{gather}
\left(\partial_U \pm \partial_V\right) \mO - \mO \left( \partial_U \pm \partial_V \right) f = 0\,.
\end{gather}
If $\Theta_U \pm \Theta_V$, the first term vanishes by the item~(\ref{dchitheta}). Therefore, if $\mO \neq 0$, then
\begin{equation}
 \left( \partial_U \pm \partial_V \right) f = 0\,.
 \end{equation}
 \item We denote by $\chi_{\pm}^a = \frac{\partial x^a}{\partial\chi^\pm}$, the components of $\partial_\pm$ in the coordinates system $x^a$. Then
\begin{gather}
\Lie_{\chi_\pm} \, g_{ab} = \chi_{\pm}^c \partial_c g_{ab}+ g_{ac}\partial_b \chi_{\pm}^c + g_{cb}\partial_a \chi_{\pm}^c = \partial_{\pm} g_{ab}=0\,,
\end{gather}
as the functions in the metric components, namely $f(U,V)$ and $r(U,V)$, do not depend on $\chi_\pm$.
\end{enumerate}
\end{proof}

In the next proposition we relate the classification of the spacetime region with the character of the Killing field.

\begin{proposition} \label{prop:killing-character}
If 
\begin{enumerate}[i.]
\item $\Theta_{(U)} = -\Theta_{(V)} \neq 0$, $\partial_+$ is a timelike Killing vector field.\label{item:timekilling}
\item $\Theta_{(U)} = \Theta_{(V)} \neq 0$, $\partial_-$ is a spacelike Killing vector field.\label{item:spacekilling}
\end{enumerate}
\end{proposition}
\begin{proof}
\begin{enumerate}[i.]
\item From Proposition~\ref{prop:killingvector}, in this case $\partial_+ = \partial_U + \partial_V$ is a Killing vector field. Then
\begin{equation}
g_{ab} \chi_+^a \chi_+^b = -2e^f < 0\,.
\end{equation}
\item Analogously:
\begin{gather}
g_{ab} \chi_-^a \chi_-^b = 2e^f > 0\,.
\end{gather}
\end{enumerate}
\end{proof}

Propositions~\ref{prop:killingvector} and ~\ref{prop:killing-character} 
imply that spherically symmetric vacuum spacetimes are also static \footnote{It 
is straightforward to show that $\chi_{\pm}^a$ is hypersurface orthogonal. In 
particular, we have $\ud \chi_a^{\pm} =   \chi_a^\pm \wedge \ud f $.}, provided 
the region is \emph{regular} or \emph{untrapped}, that is, the null expansions 
have opposite sign.  Case~(\ref{item:spacekilling}) shows that in \emph{trapped 
regions}, where both null expansions have the same sign, an additional Killing 
vector field still exists, but it is spacelike and the region is not static, 
but spatially homogeneous.

Our construction also shows that the Killing field is always orthogonal to the orbits $O_2$ and the isometry it generates commutes with the $O_2$ rotations.

\section{Beyond Vacuum and Spherical Symmetry}\label{Sec:Beyond}

We used the dual formalism under the hypothesis that the spacetime is spherically symmetric such that we used a codimension-two foliation of the spacetime using the spheres corresponding to the orbits of the action of SO(3). We also assumed that it was a vacuum solution.

Since Bona proved the Birkhoff theorem under weaker hypotheses \cite{Bona(1988)}, in this section we explore better the conditions necessary in order to prove it under our formalism.

In the proofs above, the only relevant equations were the Raychaudhuri equation for the null congruences. If we can weaken the hypotheses while keeping Eqs.~\eqref{eq:thetau-equation} and \eqref{eq:thetav-equation} unchanged, we will obtain a stronger version of our result.

\subsection{Discussing the symmetry condition}

The consequence of spherical symmetry in Raychaudhuri equations is the fact that the shear and vorticity of the null congruences must vanish, which implies that the evolution of the expansions depends only on themselves. 

In order to guarantee the vanishing of the shear and vorticity of null congruences, we can replace spherical symmetry by any maximal symmetry for two-dimensional manifolds. Therefore, we may replace the hypothesis of spherical symmetry with the statement that  $h_{\alpha \beta}$ must have constant curvature, which includes planar and hyperbolic symmetries on $h_{\alpha \beta}$. The most general line element that preserves our equations is
\begin{gather}
\ud s^2 = - e^{f(u,v)} \left(\ud u \, \ud v + \ud v \, \ud u\right) + r^2(u,v)\, h_{\alpha \beta}\ud y^\alpha \ud y^\beta\,,\label{eq:generalmetric}
\end{gather}
where
\begin{gather}
h_{\alpha \beta}\ud y^\alpha \ud y^\beta = \ud \theta^2 + S_{\epsilon}^2(\theta) \ud \phi^2\,,
\end{gather}
where $\theta \in (0, \infty)$ and 

\begin{itemize}
\item $S_{1}(\theta) = \sin \theta$, for spherical symmetry;
\item $S_0(\theta) = \theta$, for planar symmetry; and
\item $S_{-1}(\theta) = \sinh \theta$, for hyperbolic symmetry. 
\end{itemize}

This change leaves Eqs.~\eqref{eq:Rayuu} and \eqref{eq:Rayvv} invariant, while Eq.~\eqref{eq:Guv} becomes
\begin{gather}
G_{uv} = \Lie_v \Theta_u + \Theta_u \, \Theta_v + \frac{\epsilon e^f}{r^2} \,, \label{eq:Guvepsilon}
\end{gather}
where $\epsilon$ take the values 1, 0, or -1, corresponding to the spherical, planar or hyperbolic symmetry, respectively.

This is equivalent to Bona's wording in terms of the $G_3$ group of symmetry with two-dimensional orbits $O_2$, with the difference that in Bona's paper\cite{Bona(1988)}, the only exigence on $O_2$ is that it is non-null. In our case, since we use dual null basis on $V_2$, $O_2$ must be spacelike. If the orbits $O_2$ are Lorentzian, the orthogonal vector space to the orbits is spacelike; therefore, the optical focusing equations we are using cannot be applied. In this sense, our formalism is less general than Bona's.

We could use a formalism similar to prove the Birkhoff theorem in that case, by using the corresponding focusing equations for spacelike geodesics. However, as spacelike geodesics are much less interesting under the physical point of view than the null cones, and one of the objectives of this work is to discuss the physical meaning of the hypotheses of the Birkhoff theorem, we will not pursue in this direction.

\subsection{Discussing the vacuum condition}

The vacuum condition has the only effect of making Raychaudhuri equations homogeneous, since $R_{uu} = R_{vv} = 0$.

Therefore, we should determine the broadest class of energy-momentum tensors -- or, equivalently, Ricci tensors -- that produces the same result.

Since $\partial_u \equiv u^a \partial_a$ and $\partial_v \equiv v^a \partial_a$ are null, the vanishing of the $uu$ component of the Ricci tensor is equivalent to $R_a^b u^a= \lambda_u (u,v) u^b $ and analogously to $\partial_v$. Therefore, $\partial_u$ and $\partial_v$ are two null eigenvectors of the Ricci tensor. Since,  $g_{uv} = u^a v_a = -e^f \neq 0$, their respective eigenvalues $\lambda_u$ and $\lambda_v$ coincide:
\begin{gather}
\lambda_v (v^b u_b )= (R_a^b v^a) u_b  = \nonumber\\ (R_a^b u^a )v_b = \lambda_u u^b v_b\,.
\end{gather}

If the two null basis vectors are eigenvectors with the same eigenvalue, we have that $\partial_u \pm \partial_v$ also are eigenvectors with $\lambda_u$ as their eigenvalue. This shows that the condition of vanishing $R_{uu}$ and $R_{vv}$ is equivalent to  imposing that the Ricci tensor have a timelike and a spacelike eigenvector in $V_2$, with the same eigenvalue.

As the induced metric in the symmetric orbits is of constant curvature, this implies that the restriction of $R_{ab}$ to the subspace tangent to the orbits is proportional to the metric itself, that is:
\begin{gather}
R_{AB} = R(u,v) g_{AB}\,,
\end{gather}
for $A, B \in \left\lbrace \theta , \phi\right\rbrace$. This implies that the Ricci tensor has two linearly independent eigenvectors $w^a$ and $z^a$, with the same eigenvalue. The space spanned by $w^a$ and $z^a$ is orthogonal to $\partial_u$ and $\partial_v$, therefore  tangent to the orbits of the angular coordinates. This means we have $R_{uu} = R_{vv}$ for Ricci tensors of the Segr\'e type [(1,1)(11)] (two pairs of double eigenvalues), or [(111,1)] (one quadruple eigenvalue). This is the same hypothesis for the Ricci tensor used in the generalized version by Bona.

\subsection{Admissible matter models}

By  Einstein's equations, the Segr\'e type of the Ricci tensor corresponds to the Segr\'e type of the energy-momentum tensor. Therefore, it is worth determining the most general matter model that satisfies the requirements for the application of the Birkhoff theorem.

The most general $T_{ab}$ with two pairs of double eigenvectors and presenting the symmetry requirements may be written as
\begin{gather}
T_{ab} =  \lambda_1 r^2 \gamma_{ab} + \lambda_2 r^2 h_{ab} \Rightarrow \nonumber\\
T_{ab} = \lambda_1 \left(-2e^f \partial_{(a} u \, \partial_{b)} v \right) + \lambda_2 \left(r^2\,h_{ab}\right).
\end{gather}
Defining a new basis
\begin{gather}
n^a = \frac{e^{-f/2}}{2}\left[ u^a + v^a\right]\,, \nonumber\\
e^a = \frac{e^{-f/2}}{2}\left[ u^a - v^a\right]\,,
\end{gather}
which satisfy
\begin{gather}
n^a n_a = -1 \,, \quad e^a e_a = 1 \, ,
\end{gather}
we have
\begin{gather}
T_{ab}n^a n^b = - \lambda_1 \,,\quad
T_{ab}e^a e^b = \lambda_1,
\end{gather}
which leads to
\begin{gather}
T_{ab} = -\lambda_1 \, n_a n_b + \lambda_1 \, e_a e_b + \lambda_2 \, r^2 h_{ab}\,. \label{eq:emtensor}
\end{gather}

We may interpret Eq.~\eqref{eq:emtensor} as the energy-momentum tensor of a fluid with energy density $-\lambda_1$ and anisotropic pressure, with value $\lambda_1$ in the direction orthogonal to the orbits $O_2$ and value $\lambda_2$ tangent to it. If we apply the weak energy condition, then $\lambda_1 < 0$, which means that the fluid must have  \emph{negative pressure} in the $e^a$ direction.

{An important feature of the energy-momentum tensor in 
\eqref{eq:emtensor} is that the flow velocity $n^a$ is not uniquely defined, as 
a boost transformation of the form
\begin{gather}
n'^a = \cosh \omega \, n^a + \sinh \omega \,e^a\,,\nonumber\\
e'^a = \cosh \omega \, e^a + \sinh \omega \, n^a \,,
\end{gather}
for arbitrary $\omega$ preserves its form. This means that there exists a 
one-parameter family of observers, with different velocities, that are 
"comoving" to the fluid. This is a vacuumlike property, and in 
Ref.\cite{Bronnikov:2016dhz} this kind of fluid is called Dminikova vacuum, or 
D-vacuum.}

A particularly simple realization of matter of this form corresponds to $\lambda_1 = \lambda_2 = - \Lambda$, where the energy-momentum tensor correspond to a cosmological constant (in this case, the Segr\'e type is [(1,111)]).

Another case of interest is the presence of a non-null electromagnetic field 
$F_{ab}$. In the absence of charges and radiation 
\cite{Kopczynski-Trautman-book} the energy-momentum tensor may be written as
\begin{gather}
T_{ab} = \frac{1}{2}\left(E^2 + B^2\right) [n_a n_b - e_a e_b + r^2 h_{ab} ]\,,
\end{gather}
which corresponds to a Segr\'e type $[(1,1)(11)]$, for $\lambda_2 = -\lambda_1 
=\frac{E^2 + B^2}{2}$.

We see that the most known cases where the Birkhoff theorem is usually applied 
in the literature are quite particular, as they correspond to $\lambda_1 = \pm 
\lambda_2$. In the next section, we will solve  Einstein's equations for a general matter model satisfying the above conditions.

In general, we may represent an energy-momentum tensor of the Segr\'e types required as
\begin{gather}
T^{ab} = \lambda_1 g^{ab} + \left(\lambda_2 - \lambda_1 \right) r^{-2} h^{ab}\,.
\end{gather}
The energy-momentum conservation is written
\begin{gather}
0 = \nabla_a T^{ab} 
= g^{ab}\partial_a \lambda_1 + r^{-2}h^{ab} \partial_a \left(\lambda_2 - \lambda_1 \right) + \nonumber\\ (\lambda_2 - \lambda_1) \nabla_a (r^{-2} h^{ab} )\,.\label{eq:conservation}
\end{gather}

If $\lambda_2 = \lambda_1$, this implies $\partial_a \lambda_1 = 0$, which means 
that the eigenvalue must be constant, as  is well known for the cosmological 
constant. Therefore, $\lambda_2 \neq \lambda_1$ is necessary in order to obtain 
models with varying $\lambda_1$. We will see in the next section that only 
$\lambda_1$ appears directly in Einstein's equations, but $\lambda_2$ affects 
implicitly the solution as it is related to $\lambda_1$ according to 
Eq.~\eqref{eq:conservation}.

{A thorough presentation of the field Lagrangians that 
produce this type of energy-momentum tensors can be found in Ref.
\cite{bronnikov1980generalisation}. {Another important remark concerns interpretation of the results in terms of matter models: }
This analysis is 
equally valid for extra terms in Einstein equations provided by modified gravity 
theories.}

\subsection{Summarizing our results}

We conclude this section by stating the generalized version of the Birkhoff theorem in our language.

\begin{theorem}
Let $g$ be a metric tensor of a spacetime that admits a codimension-two foliation of the form Eq.~\eqref{eq:generalmetric} where $h$ is two-dimensional Riemannian metric tensor, induced on the two-dimensional spacelike leaves of the foliation. 

If $h$ has constant curvature, $\Theta_u \neq 0$, as defined in Eq.~\eqref{eq:nullexpansion} and the energy-momentum tensor has the form given in Eq.~\eqref{eq:emtensor}, then $g$ has an additional isometry generated by a Killing vector $\chi$ orthogonal to the leaves of the foliation. 

In addition, if the spacetime region considered is \emph{regular}, then $\chi$ is timelike and the metric is static. If the spacetime region is \emph{trapped} or  \emph{antitrapped}, then $\chi$ is spacelike and the metric is homogeneous.
\end{theorem}

\section{Solving the equations}

In this section we aim to determine the solutions that satisfy the theorem, by 
using the tools we have already prepared. We have to consider two types of 
solutions: those for regular or untrapped regions, where $\Theta_U + \Theta_V~ 
= 0$ and those for trapped regions, corresponding to $\Theta_U - \Theta_V = 0$.

\subsection{Regular regions}

With no loss of generality,  we assume $\mO = \Theta_U > 0$ and $\Theta_V 
= - \mO < 0$.

It is useful to remark that, in this case $\partial_+ r = 0$, which means that the vectors $\partial_+$ and $\partial_r$ are orthogonal, which implies that $\partial_-$ is proportional to $\partial_r$. Indeed, it is straightforward to verify that
\begin{gather}
\mO= \frac{1}{r}\partial_- r \Rightarrow \partial_- = r \mO \partial_r\,, \label{eq:delminus}
\end{gather}
wherever $\partial_- r \neq 0$.

We must revisit Eq.~\eqref{eq:UVequations} and note that by redefining $f$ using the transformation in Eq.~\eqref{eq:Fredefined}, all the functions on the right-hand side are in the exponential term. Since the coordinates $U$ and $V$ are unique up to a rescaling transformation we can set the proportionality constant as $ 2$, and then
\begin{gather}
r\mO =  2e^f\,, \label{eq:expf}
\end{gather}
which allows us to write the line element as
\begin{gather}
\ud s^2 = - \frac{r\mO}{2}\left(\ud U \, \ud V + \ud V \, \ud U \right) + r^2 h_{\alpha \beta}\ud x^\alpha \, \ud x^\beta\,,\label{eq:metricUV}
\end{gather}

Also, according to Proposition~\ref{prop:killingvector}, no metric component depends on $\chi^+$, and its basis vector is orthogonal to $\partial_r$, which makes the pair of coordinates $(\chi^+, r)$ a natural choice to describe the solution. Using Eq.~\eqref{eq:expansion-definition} and the definition of $\chi^+$, we obtain

\begin{gather}
\ud r = \frac{r\mO}{2} \left( \ud U - \ud V \right)\,,\nonumber\\
\ud \chi^+ = \frac{1}{2}\left(\ud U + \ud V\right)\,,\label{eq:coordstransformation}
\end{gather}
which leads to
\begin{gather}
\ud s^2 = -{r\mO} \, \ud \chi^{+2}+\frac{\ud r^2}{r \mO}+ r^2 \, h_{\alpha \beta}\,\ud x^\alpha \, \ud x^\beta\,.\label{eq:metricsolution}
\end{gather}

Now, we have only to solve the equations for $\mO(r)$. Considering a matter model of the form Eq.~\eqref{eq:emtensor}, applying the Einstein equations Eq.~\eqref{eq:einstein} to Eq.~\eqref{eq:Guv} and considering Eq.~\eqref{eq:expf}, we have
\begin{gather}
\frac{1}{2} \left(\partial_V \Theta_U + \partial_U \Theta_V \right) + \Theta_U \Theta_V + \epsilon \frac{e^f}{r^2} = -\lambda_1 e^f \Rightarrow 
\nonumber\\
\frac{1}{2} \left(\partial_V - \partial_U \right) \mO - \mO^2 + \frac{\mO}{2}\left( \frac{\epsilon}{r} + \lambda_1 r \right)=0 \Rightarrow \nonumber\\
-\frac{1}{2} \partial_- \mO - \mO^2 + \frac{\mO}{2} \left(\frac{\epsilon}{r} + \lambda_1 r\right)=0\,.\label{eq:delminusO}
\end{gather}

 Using Eq.~\eqref{eq:delminus}, we are able to write a differential equation with respect to $r$:
 \begin{gather}
 r \partial_r \mO + 2\mO =  \left(\frac{\epsilon}{r} + \lambda_1 (r) r \right) \Rightarrow \nonumber\\
  \partial_r \left( r^2 \mO \right) = \left({\epsilon} + \lambda_1 (r) r^2 \right) \Rightarrow \nonumber\\
  r \mO = \epsilon + \frac{b}{r} + \frac{1}{r}\int^r \lambda_1(s)s^2 \ud s .\label{eq:rO-solution}
 \end{gather}
 
 In order to gain a better insight on our family of solutions, we consider the 
case where  $\lambda_1(r)$ admits a representation as a sum or a series of 
powers of $r$, provided it is uniformly convergent on $\mathcal{D}$:
 \begin{gather}
 \lambda_1(r) =- \sum_i c_i r^i\,,
 \end{gather}

Then, we can integrate it term by term and find
\begin{gather}
r\mO = \epsilon + \frac{b}{r} - \frac{c_{-3}\ln r}{r} - \sum_{i \neq -3} \frac{c_ir^{i+2}}{i+3}\,. \label{eq:generalsolution}
\end{gather}

The most common sources studied in black hole physics are particular cases of our model. The cosmological constant is equivalent to $c_0 =  \Lambda$, the electrostatic central field to $c_{-4} =  q^2$. We also see that in this case the solutions behave  "linearly": The addition of sources, provided they satisfy the requirements of the Birkhoff theorem, corresponds to the addition of a respective term in the solution.{It is worth noticing that in the Newtonian limit, $r\mO \sim 1 + 2\Phi$, where $\Phi$ is the  potential and Eq.~\eqref{eq:rO-solution} is the relativistic analog of the Poisson equation. As in the Newtonian case, the solution~\eqref{eq:generalsolution} must satisfy boundary conditions at the innermost and outermost radius of the domain $\mathcal{D}$, including $r\to 0$ and $r \to \infty$ as possible cases. However, as explicitly shown in Ref.\cite{Kim:2016osp}, the boundary conditions at the outermost radius may affect physics in $\mathcal{D}$, in opposition to the result of Newtonian gravity.}

Another fact of interest is that all those solutions are of Petrov type D, 
meaning that the geometry corresponds only to the Coulombian part of the 
gravitational field. This is expected, since the high degree of symmetry of those 
solutions eliminates any form of gravitational radiation term.

For the spherical solutions, we are able to compute the Misner-Sharp mass \cite{Hayward:1994bu} of the general solution as
\begin{gather}
M = \frac{r}{2} \left(1 - g^{ab}\partial_a r \, \partial_b r\right)\,.
\end{gather}
From the line element in Eq.~\eqref{eq:metricsolution}, we have $ g^{ab}\partial_a r \, \partial_b r = r\mO$. Using the general solution Eq.~\eqref{eq:generalsolution}, we have
\begin{gather}
M= \frac{1}{2}\left( -b +{c_{-3}\ln r} + \sum_{i \neq -3} \frac{c_ir^{i+3}}{i+3} \right)\,,
\end{gather}
which allows us to identify that the integration constant $b = -2m$, where $m$ is the central mass, as it is the component of the total Misner-Sharp mass which is independent of the radius and corresponds to the Schwarzschild mass in the absence of sources. The other terms give the energy contribution of each kind of source.

\subsection{Trapped 
and antitrapped regions}

Those regions correspond to $\Theta_U = \Theta_V = \mO$. Trapped regions present $\mO < 0$ and antitrapped regions have $\mO~>~0$. We follow the changes in the equations we presented for regular regions.
 In this case $\partial_- r = 0$, and then we replace
 Eq.~\eqref{eq:delminus} by
 \begin{gather}
 \frac{1}{r} \partial_+ r = \mO\, \Rightarrow \partial_+ = r\mO \partial_r,\label{eq:partialtrapped}
\end{gather}

Equation~\eqref{eq:expf} becomes
\begin{gather}
r\mO = \pm 2 e^f\,,
\end{gather}
because we need to include the case where $\mO < 0$.
The line element in Eq.~\eqref{eq:metricUV} becomes
\begin{gather}
\ud s^2 = - \frac{r|\mO|}{2}\left(\ud U \, \ud V + \ud V \, \ud U \right) + r^2 h_{\alpha \beta}\ud y^\alpha \, \ud 
y^\beta\,.\label{eq:metricUVtrapped}
\end{gather}

We replace Eqs;~\eqref{eq:coordstransformation} by
\begin{gather}
\ud r = \frac{r\mO}{2} \left( \ud U + \ud V \right)\,,\nonumber\\
\ud \chi^+ = \frac{1}{2}\left(\ud U - \ud V\right)\,,
\end{gather}

The line element in the coordinates $(\chi^- , r )$ is given by
\begin{gather}
\ud s^2 = {r|\mO|} \, \ud \chi^{-2}-\frac{\ud r^2}{r |\mO|}+ r^2 \, h_{\alpha \beta}\,\ud y^\alpha \, \ud y^\beta\,.\label{eq:metricsolutiontrapped}
\end{gather}

We notice that in trapped regions, the coordinate $r$ is timelike and the corresponding metric element is negative, as expected from Proposition~\ref{prop:killing-character} along with the fact that $\partial_- r = 0$.

Equation~\eqref{eq:delminusO} becomes
\begin{gather}
\frac{1}{2} \partial_+ \mO + \mO^2 \pm \frac{\mO}{2} \left(\frac{\epsilon}{r} + \lambda_1 r\right)=0\,, \label{eq:delminusOtrapped}
\end{gather}
where the $+$ sign correspond to antitrapped regions and the $-$ sign to trapped regions. Changing the $\chi^+$ coordinate to $r$, using Eq.~\eqref{eq:partialtrapped}, we obtain
\begin{gather}
\frac{1}{2} r \partial_r \mO + \mO \pm \left(\frac{\epsilon}{r} + \lambda_1 r\right)=0\,, \Rightarrow \nonumber\\
r\mO = \mp \left(\epsilon + \frac{b}{r} + \frac{1}{r}\int^r \lambda_1(s)s^2 \ud s   \right)\,.
\end{gather}

Therefore, the only difference in the case of trapped and antitrapped regions lies in the character of the Killing vector and of the $r$ coordinate. The absolute value of the metric components coincide.

Notice that the metric solutions we found fail to cover the marginal surfaces that correspond to $\mO = 0$. The three-dimensional locus defined by the marginal surfaces is an apparent (or trapping) horizon \cite{Hayward:1993mw}, which is the boundary between trapped and untrapped regions. While our choice of coordinates $(\chi^\pm , r)$ makes use of the symmetry of the problem in order to simplify its resolution, 
the metrics in Eqs.~\eqref{eq:metricsolution} and \eqref{eq:metricsolutiontrapped}
have the Schwarzschild form in usual coordinates, and the marginal surfaces correspond to coordinate singularities. A coordinate system that covers both sides of those marginal surfaces is easily built by  known methods as, for instance, the definition of an Eddington-Finkelstein-like system of coordinates. This analysis leads to the known fact that the Killing field is null on a marginal surface.

\section{Conclusion}\label{Conclusion}

We have shown how to obtain the Birkhoff theorem from the dual null formalism, naturally relating the result with the type of region considered, if regular or trapped. Only in a regular region does the theorem lead to static solutions.

The formalism has also enabled us to prove a very general version of the Birkhoff 
theorem, coming short of being completely general in that we did not consider 
symmetries with timelike orbits as done by Bona \cite{Bona(1988)}. However, we 
have {obtained} general matter sources for which the theorem is valid and 
thus, with the aid of dual null formalism, we found all the solutions for 
sources that can be expressed as a power series on $r$.

The Birkhoff theorem is much invoked in the literature in relation to the idea that 
given a spherically symmetric distribution of matter, the gravitational physics 
at some given value of the  radial coordinate depends only on the overall mass 
of the distribution inside that radius. This is, of course, not true in general, 
and a clear counterexample is provided by the well-known 
Lema\^{\i}tre-Tolman-Bondi dust solution \cite{Plebanski:2006sd}, in which the 
gravitational physics, at some spherical shell, depends not only on the 
integrated Misner-Sharp mass but also on an energy parameter that weights the 
spatial curvatures and the initial energy conditions. Other misuses have been 
discussed in Ref. \cite{Kim:2016osp}.  We thus believe that the present work is 
transparent and useful in making it absolutely clear what is the scope of  
applicability of the Birkhoff theorem in general relativity, and also as guide for 
the investigation of analogous results in modified gravities theories.

\begin{acknowledgments}
{The authors wish to thank K.A. Bronnikov for drawing our 
attention to useful references in literature.} 
A. M. thanks Conselho Nacional de Desenvolvimento Cient\'ifico e Tecnol\'ogico (CNPq), Brazil, Grant n\textordmasculine~400342/2017-0. JPM  acknowledges the financial  support by Funda\c{c}\~ao para a Ci\^encia e a Tecnologia (FCT) through the research grant UID/FIS/04434/2013. MLeD acknowledges the financial  support by Lanzhou University starting fund and wish to thank the hospitality of Instituto de Astrof\'{\i}sica e Ciencias do Espa\c co (IA), at the FCUL in Lisbon, where a part of this work was carried out. {MLeD and JPM are most grateful to Tera Shimizu for the tasty discussions and musical contribution to the progress of our work.}
\end{acknowledgments}

 \bibliography{referencias, shortnames}

\end{document}